\documentclass[a4paper,11pt,reqno]{article}
\usepackage[utf8]{inputenc}
\usepackage{subcaption}
\usepackage{tabularx}
\usepackage{amsmath,amssymb,amsthm}
\usepackage[english]{babel}
\usepackage{mathtools}
\usepackage{enumerate}
\usepackage{hyperref}
\usepackage[numbers]{natbib}
\usepackage{comment}
\usepackage{graphicx}
\usepackage[dvipsnames]{xcolor}
\usepackage{color}
\usepackage{tikz}
\usetikzlibrary{calc,fit,matrix,arrows,automata,positioning}

\usepackage{hyperref}
\hypersetup{
  colorlinks   = true, 
  urlcolor     = blue, 
  linkcolor    = OliveGreen, 
  citecolor   = red 
}

\usepackage{tabularx}
\usepackage{booktabs}
\usepackage{makecell}
\usepackage{longtable}
\usepackage[ruled,vlined]{algorithm2e}
\usepackage{algorithmic}
\usepackage{listings}
\usepackage{etoolbox}
\usepackage{lipsum}
\usepackage[left =2.6cm, right = 2.6 cm]{geometry}

\usepackage{url}

\usepackage{amssymb}
\usepackage{amsfonts}
\usepackage{amsmath}
\usepackage{comment}
\usepackage{tabularx}
\usepackage{booktabs}
\usepackage{makecell}
\usepackage{longtable}
\usepackage{graphicx}
\usepackage{caption}
\usepackage{lipsum}
\usepackage{xcolor}
\usepackage{parskip}
\usepackage{caption}
\usepackage{floatrow}

\usepackage{hyperref, enumitem}

\newcommand{\Z}{\mathbb{Z}}
\newcommand{\Zn}{\mathbb{Z}_n}
\newcommand{\F}{\mathbb{F}}

\newfloatcommand{capbtabbox}{table}[][\FBwidth]

\newtheorem{theorem}{Theorem}

\newtheorem{lemma}{Lemma}

\newtheorem{corollary}{Corollary}
\theoremstyle{definition}

\newtheorem{definition}{Definition}
\theoremstyle{remark}
\newtheorem{remark}{Remark}
\newtheorem{example}{Example}

\usepackage[english]{babel}

\usepackage{amsmath}
\usepackage{graphicx}

\title{A Number Theoretic Approach to Cycles in LDPC Codes}
\author{Julia Lieb \and Simran Tinani}

\begin{document}
\maketitle

\begin{abstract} 
LDPC codes constructed from permutation matrices have recently attracted the interest of many researchers. A crucial point when dealing with such codes is trying to avoid cycles of short length in the associated Tanner graph, i.e. obtaining a possibly large girth. In this paper, we provide a framework to obtain constructions of such codes. We relate criteria for the existence of cycles of a certain length with some number-theoretic concepts, in particular with the so-called Sidon sets. In this way we obtain examples of LDPC codes with a certain girth. Finally, we extend our constructions to also obtain irregular LDPC codes.
\end{abstract}

\section{Introduction}


LDPC (Low-density parity-check) codes were first introduced in the thesis of Gallager \cite{gall} and started gaining traction for practical use when they were rediscovered by MacKay \cite{mackay}. These codes have recently been attracting the interest of many researchers, since bit-flipping and message-passing decoding algorithms allow to decode this type of codes very efficiently. They have even been proven to be able to reach the famous Shannon capacity limit \cite{mackay}. With the codes' sparse parity-check matrix one can associate a sparse bipartite graph called the Tanner graph \cite{tanner}. 
The performance of iterative decoding
algorithms on these codes has been shown to depend strongly on the structure of the associated Tanner graph, and in particular, on the presence of short cycles in it. For an exposition on this, see, for example, \cite{girth}, \cite{lucas}, and \cite{mck2}.
Therefore, it is desirable to come up with constructions of codes where the size of the smallest cycle in the graph, the so-called girth, is not too small.

One subclass of LDPC codes that has been considered in this direction of research are LDPC codes built from permutation matrices, i.e. LDPC codes whose parity-check matrix is a block matrix, where each constituent block is given by a permutation matrix. Besides considering general permutation matrices as building blocks \cite{perm}, some papers focus on special permutations, such as permutations given by quadratic polynomials \cite{quadratic}, linear polynomials \cite{linear}, or circulant permutations \cite{qcldpc}, \cite{ma-lee}. If one uses circulant permutation matrices as building blocks, one obtains the so-called quasi-cyclic LDPC codes (see \cite{Liva2017DesignOL} for a survey). This class of quasi-cyclic LDPC codes has been most extensively studied as their circulant structure allows for easy implementation. However, it is known that the performance of LDPC codes built from circulant permutation matrices is not optimal (see e.g. \cite{partially}) as their girth and minimum distance are limited by an upper bound.

In \cite{unifying}, a general criterion to determine the girth of an LDPC code whose parity-check matrix consists of permutation matrices was presented, building on earlier work in \cite{mcgowan} and \cite{wu}. Also in \cite{unifying}, this criterion was further simplified and studied in more detail for the case that the parity-check matrix consists of circulant permutation matrices.

In this work, we study parity-check matrices consisting of permutation matrices that are not necessarily circulant but are all powers of a fixed permutation f. With this choice, we relate the presence of cycles of a certain length in the corresponding Tanner graph to some additive number theoretic properties of the set of powers of $f$ which appear in the matrix. In particular, we find interesting connections with the concept of Sidon sets, introduced in \cite{Sidon1932EinS}, and their modular variations. In this way, we provide new constructions of LDPC codes with easy to determine girth. In fact, for these constructions, we give explicit formulations of equivalent conditions for 4- and 6- cycles, and a necessary condition for 8-cycles, exclusively in terms of the arithmetic properties of the power sets. For the special case with only two rows of block matrices, 6-cycles are impossible, and we provide if and only conditions for 4 and 8-cycles, thus allowing the construction of girth 12 matrices. 

We remark that our results contain quasi-cyclic LDPC codes as a special case. Indeed, any quasi-cyclic LDPC code consists of blocks, each corresponding to a permutation of the form $i \mapsto i+ k$ for some fixed value of $k$. 
Clearly, every such permutation can be obtained as a power of the cyclic permutation $f:\mathbb Z_n\rightarrow\mathbb Z_n, \ i\mapsto i+1$.

The paper is structured as follows. In Section \ref{fossorier}, we recall some theory from \cite{foss} providing criteria for cycles in the Tanner graph of a certain length. In Section \ref{conc}, we define some concepts from number theory that we will need for later sections of the paper. In Section \ref{construction}, we introduce our code construction together with some general criterion for cycles in the Tanner graph of a certain length. Then, in Section \ref{cycles}, we give more explicit conditions to avoid $4$, $6$ and $8$-cycles relating them to Sidon sets. Finally, in Section \ref{irregular}, we explain how our code constructions can be extended to obtain irregular LDPC codes with the same girth but higher code rates. We conclude with some final remarks in Section \ref{concl}.

Throughout the paper, for any permutation we will use the notation as a function or via the corresponding permutation matrix interchangeably.

\section{Fossorier condition}\label{fossorier}

As starting point, we use the theory developed in \cite{foss} by Fossorier, which introduces another way to obtain criteria for cycles in the Tanner graph of an LDPC code.

Let $H$ be a block parity check matrix of the form $\begin{pmatrix} P_{0,0} & \ldots & P_{0,J-1} \\
\vdots & \ddots & \vdots \\
P_{L-1,0} & \ldots &  P_{L-1, J-1}
\end{pmatrix}.$
Each $P_{l,j}$ is an $n\times n$ permutation matrix defined as follows by a permutation $R_{i,j}:\Z_n \rightarrow \Z_n$. \[P_{l,j}(i,k) = 1 \ \text{if} \ i =R_{l,j}(k)\] 

We will write the entries of $P_{j,l}$ as $h_{jm+j', lm+l'}$, $0\leq j' \leq J-1$, $0 \leq l' \leq L-1$. Then, $h_{jm+j', lm+l'}=1 \iff l' = R_{(j,l)}(j')$.

For simplicity, we will write $H$ in terms of the permutations, $H=\begin{pmatrix} R_{0,0} & \ldots & R_{0,J-1} \\
\vdots & \ddots & \vdots \\
R_{L-1,0} & \ldots &  R_{L-1, J-1}
\end{pmatrix}$.

\begin{remark}
As row permutations inside $H$ do not change the code and column permutations inside $H$ lead to an equivalent code, one can assume that $R_{i,j}=id$ if $i=0$ or $j=0$. Thus, henceforth we will deal with parity check matrices of the form 
$$H=\begin{pmatrix} Id_{n\times n} & Id_{n\times n} & \ldots &Id_{n\times n}  \\
Id_{n\times n} & R_{1,1} & \ldots & R_{1,L-1} \\
\vdots & \vdots & \ddots & \vdots \\
Id_{n\times n} & \ldots &  R_{L-1, J-2} & R_{L-1, J-1}
\end{pmatrix}.$$

\end{remark}

For simplicity, we will often write 1 instead of $Id_{n\times n}$.

\begin{definition}[Cycle]
A cycle in $H$ is a sequence of positions $(x_i, y_i)$ of the form $h_{x_0, y_0}=1, h_{x_1, y_0}=1, h_{x_1, y_0}=1, \ldots, h_{x_{k-1}, y_{k-1}}=1, h_{x_{0}, y_{k-1}}, h_{x_0, y_0}$ in $H$. In other words, a cycle is defined by positions with entry 1, obtained by changing alternatively row or column only, and such that all positions are distinct, except the first and last. Such a cycle is said to have length $2k$.
\end{definition}

Since the $P_{i,j}$'s are permutation matrices, we must have each entry $h_{x_i, y_i}$ in a different permutation block matrix. In this paper, we follow the convention in \cite{foss}, so that the first change is of a row, followed by a column. Consequently, each cycle can be associated to a unique path $(j_0,l_0) \rightarrow \ldots \rightarrow (j_{k-1}, l_{k-1}) \rightarrow (j_k, j_{k-1}) \rightarrow (j_k, l_k)=(j_0, l_0)$ of indices labelling the permutation matrices, $0\leq j_i \leq J-1, 1\leq l_i \leq L-1$.

\begin{definition}[Girth]
The girth of $H$, denoted $g(H)$, is defined to be the smallest positive even integer $2k$, $k>1$, such that $H$ contains a $2k$-cycle.
\end{definition}

\begin{theorem}[\cite{foss}]
There exists a cycle associated to the sequence $$(j_0,l_0), \ldots, (j_{k-1}, l_{k-1}), (j_{0}, j_{k-1}), (j_k, l_k)=(j_0, l_0)$$ if and only if there exists a column index $c_0$, $0\leq c_0 \leq l_0$, 
such that $R_{j_0,l_0}(c_0) =  R_{j_0, l_{k-1}}(c_{k-1})$, or $c_k = c_0$, where $c_{i+1} = (R_{j_{i+1}, l_{i+1}}^{-1}(R_{j_{i+1}, l_i}(c_{i}))$ for $0 \leq i \leq k-1$.
\end{theorem}
\begin{proof}
 Suppose that such a column index $c_0$ exists and $c_i$'s are defined by 
 $$c_{i+1} = R_{j_{i+1}, l_{i+1}}^{-1}(R_{j_{i+1}, l_i}(c_i)) \ \forall i \in \{0, 1, \ldots, k-1\}.$$ Then, clearly, 
\begin{align}\label{star1}
    j_{i+1}m+R_{j_{i+1}, l_i}(c_i) = j_{i+1}m + R_{j_{i+1}, l_{i+1}}(c_{i+1}) \forall i \in \{0, 1, \ldots, k-1\} 
\end{align}
Write $x_i =  j_{i}m+R_{j_{i}, l_i}(c_i)  $, $y_i = l_im + c_i$. Then, the entries $h_{(x_0, y_0)}, h_{(x_1, y_1)}, \ldots, h_{(x_{k-1}, y_{k-1})}, h_{(x_k, y_k)} $ are all clearly 1. Further, by \eqref{star1}, $j_{i+1}m + R_{j_{i+1}, l_i}(c_i) = x_{i+1}$, so $\forall i \in \{0, 1, \ldots, k-1\}$, $h_{(x_{i+1}, y_i)}=1$. Finally, $c_k= R_{j_k, l_k}^{-1}(R_{j_k, l_{k-1}}(c_{k-1}) = R_{j_0,l_0}^{-1}(R_{j_0, l_{k-1}}(c_{k-1})) = c_0$. Thus, the entries at the positions $${(x_0, y_0)}, {(x_1, y_1)}, \ldots, {(x_{k-1}, y_{k-1})}, {(x_k, y_k)} = (x_0, y_0)$$ form a cycle. Conversely, suppose that there is a cycle associated to the sequence $$(j_0,l_0), (j_1, l_0) \ldots, (j_{k-1}, l_{k-1}), (j_{0}, j_{k-1}), (j_k, l_k)=(j_0, l_0).$$ Then there exist positions \begin{align*}
    {(x_0, y_0)}, (x_1, y_0),  \ldots, (x_0, y_{k-1}), {(x_k, y_k)} = (x_0, y_0)
\end{align*} whose entries are all 1's, where $x_i =  j_{i}m+R_{j_{i}, l_i}(c_i)  $, $y_i = l_im + c_i$. To have $h_{(x_{i+1}, y_i)}=1$ $\forall i \in \{0, 1, \ldots, k-1\}$, we additionally need $j_{i+1}m + R_{j_{i+1}, l_i}(c_i) = x_{i+1}$, or $c_{i+1} = (R_{j_{i+1}, l_{i+1}}^{-1}(R_{j_{i+1}, l_i}(c_{i}))$. This completes the proof.
\end{proof}

The following is an easy generalization of Theorem 2.5 in \cite{foss} and provides an upper bound for the girth.

\begin{theorem}\label{12c} Let $L>2$. If $H$ has a submatrix of the form 
 $\begin{pmatrix} 1 & 1 & 1 \\ 1 & \sigma_1 & \sigma_2 \end{pmatrix}$ where $\sigma_1\sigma_2 = \sigma_2\sigma_1$ 
 then $H$ contains a 12-cycle, and so $g(H) \leq 12$.
 \end{theorem}
 \begin{proof}
 Since the $\sigma_i$'s commute, the path given by
 \begin{align*}
     1 \rightarrow 1 \rightarrow \sigma_1 \rightarrow 1 \rightarrow 1 \rightarrow \sigma_2  \rightarrow 1 \rightarrow 1 \rightarrow 1 \rightarrow \sigma_1  \rightarrow \sigma_2 \rightarrow 1 \rightarrow 1
 \end{align*} which can be visualized as  $\begin{pmatrix} 1 & 1 & 1 & 1 & 1 & 1 \\ 1 & \sigma_1 & \sigma_2 & 1 & \sigma_1 & \sigma_2 \end{pmatrix},$ defines a 12-cycle, since $c_1= \sigma_1^{-1}(c_0)$, $c_2=c_1$, $c_3= \sigma_2(c_2)= \sigma_2\sigma_1^{-1}(c_0) $, $c_4 =c_3$, $c_5= \sigma_2^{-1}\sigma_1(c_4)=\sigma_2^{-1}\sigma_1\sigma_2\sigma_1^{-1}(c_0) = (\sigma_2^{-1}\sigma_1)(\sigma_2\sigma_1^{-1})(c_0) = c_0$, $c_6= c_5=c_0$. 
 \end{proof}

 \begin{remark}
 From the above proof, we see that in order to avoid a 12-cycle, we require non-commuting permutations $\sigma_1$ and $\sigma_2$ such that $(\sigma_2^{-1}\sigma_1)(\sigma_2\sigma_1^{-1})$ is a derangement. Since we are concerned with matrices using commuting permutations in this paper, all our matrices have a maximum girth of 12.
 \end{remark}
 
 More generally, we can make the following statement.

 \begin{theorem} Let $L>2$. If $H$ has a submatrix of the form 
 $\begin{pmatrix} 1 & 1 & 1 & \ldots & 1\\ 1 & \sigma_1 & \sigma_2 & \ldots & \sigma_{2r} \end{pmatrix}$ for $r$ odd and there exists some $c \in \Z_m $  such that $\sigma_i\sigma_j (c)= \sigma_j\sigma_i(c)$ and $(\sigma_1\sigma_3\ldots \sigma_{2t-1})(c) \neq (\sigma_2\sigma_4 \ldots \sigma_{2t})(c)$ for $1\leq t<r$ and $(\sigma_1\sigma_3\ldots \sigma_{2r-1})(c) = (\sigma_2\sigma_4 \ldots \sigma_{2r})(c)$, then $H$ has a $4\cdot (2r+1)$-cycle.
 \end{theorem}


 \section{Relevant additive number theoretic concepts}\label{conc}
 
 We begin by introducing some concepts of additive number theory which will be shown later in this paper to be relevant to the theory of LDPC codes.
 
 \begin{definition}
For a set $\mathcal{B}$, we define the set $\mathcal{B}_{\Delta}$ of differences and the sum set $\mathcal{B}_+$ as follows \[\mathcal{B}_\Delta=\{a-b \mid a, b \in \mathcal{B}, a\neq b\}\]
\[\mathcal{B}_+ = \{a+b \mid a,b \in \mathcal{B}. \}\]
\end{definition}

\begin{definition}
To a set $\mathcal{B}$, we associate the following set of sequences \[\mathcal{B}_{S}^t=\{(a_0, a_1\ldots, a_{t}) \mid a_i\neq a_{i+1}, 0\leq i< t, \ a_{t}\neq a_0\}\subseteq \mathcal{B}^{t}\]
 To any sequence $B=(a_0, a_1\ldots, a_{t})\in \mathcal{B}_{S}^{t}$,  we associate the sequence \begin{align*}
     B_d = (a_0-a_{t}, a_1-a_0, \ldots, a_{t}-a_{t-1}).
 \end{align*} We also define the set $\mathcal{B}_{D}^{t}$ of all sequences $B_d$.
 \begin{align*}
\mathcal{B}_{d}^t=\{(a_0-a_t, a_1-a_0\ldots, a_{t}-a_{t-1}) \mid a_i \in \mathcal{B}, a_i\neq a_{i+1}, 0\leq i< t, \ a_{t}\neq a_0\}\subseteq (\mathcal{B}_\Delta)_S^{t}.
 \end{align*}
\end{definition}

 We now introduce a number theoretic concept originally defined in the work of Sidon on Fourier series \cite{Sidon1932EinS}. Sidon defined a ${B}_2$ sequence as a sequence $a_1 < a_2 < \ldots$ of positive integers such that the sums $a_i+a_j, \ i<j$ are all different. In accordance with this definition, in \cite{jia} and \cite{chen}, a Sidon $B_h$-sequence was defined as a set of positive integers such that all the sums $a_1+a_2+\ldots + a_h$, $a_i \in {B}_h$, are distinct, up to rearrangement of the summands. This concept was applied to sets modulo $n$ rather than over $\Z$ in \cite{chen}, where the authors called it a $B_h$ sequence for $\Zn$. Problems related to a distinct subset sums set over integers, i.e. a set $ S \subseteq {\mathbb Z} $ whose distinct subsets have distinct sums, have also been studied, for example, in \cite{erdos}, \cite{Bohman} and \cite{guy}. In this paper, we are mostly interested in $B_2$-sequences, and all arithmetic considered will be modulo $m$.

  \begin{definition}
A set $\mathcal{I} \subseteq \Z_m$ is called a Sidon $m$-$B_t$ set (or, in short, an $m$-$B_t$ set) if for $i_1, i_2, \ldots, i_t, j_1, \ldots j_t \in \mathcal{I}$, we have $i_1+i_2 +\ldots + i_t = j_1+j_2+\ldots + j_t \iff  \{i_1, i_2, \ldots, i_t \}= \{j_1, j_2, \ldots, j_t \}$.   Since all arithmetic in this paper is modulo a finite integer $m$, we sometimes omit the prefix $m$ and refer to $B_t$-sets.
 \end{definition}
 
 \begin{lemma}\label{equivconds} The following conditions are equivalent.
 \begin{enumerate}
     \item $\mathcal{I}$ is a $B_2$-set
     \item $|\mathcal{I}_+| =  (|\mathcal{I}|^2+|\mathcal{I}|)/2$
     \item $0 \not\in (\mathcal{I}_+)_\Delta$.
 \end{enumerate}
 \end{lemma}

 \begin{remark}
$B_t$-sets are in close relation to what is known as difference set in the literature. 
More precisely, an $m$-$B_2$-set $\mathcal{I}$ with $|\mathcal{I}|=k$ is equal to a $(m,k,1)$ difference set (see Definition 7.19 in \cite{lidl1997finite})

 \end{remark}

 \begin{lemma}\cite{unifying}
Set $i_1=0$ and $i_l=1+2i_{l-1}$. Then, one obtains an $m$-$B_2$ set for any odd $m$.
\end{lemma}

\section{Code Constructions: General Formulation}\label{construction}

In this section, we present a method to construct LDPC codes whose cycles are directly related to the concepts introduced in Section \ref{conc}. We denote by $\mathcal{S}_n$ the symmetric group on $n$ symbols. We will construct parity check matrices from a single permutation $f\in \mathcal{S}_n$.  Let $f: \Z_n  \rightarrow \Z_n$ be a permutation and let  $m$ be its order in the group $\mathcal{S}_n$. For compatibility of the notations of symbols in $\Z_n$ and $\mathcal{S}_n$, we use the symbols $0$ and $n$ interchangeably. For example, the permutation $(n-1, n) \in \mathcal{S}_n$ represents the map $n-1 \mapsto 0$, $0 \mapsto n-1$, $i \mapsto i \ \forall \ i\in \Z_n, i\neq 0, \; n-1$. Note that if $f$ is an $n$-cycle then $m=n$. If $n$ is prime and $f$ is an $n$-cycle, then $m=n$ and we say that $f$ is a prime $m$-cycle. A permutation $f$ is said to be a derangement if it has no fixed points, i.e. if $f(c)\neq c \ \forall \ c\in \Z_n$. 
 
 The following two results are elementary. 
 \begin{lemma} Let $f \in \mathcal{S}_n$. If $c_0 \in \Z_n$ is a fixed point of $f$ then it is a fixed point of $f^i$ for every $ i\geq 1$. 
 \end{lemma}

 \begin{lemma} Let $f \in \mathcal{S}_n$ have order $m$. If $f$ is a derangement then $f^{i}$ is a derangement for $i$ coprime to $m$. 
 \end{lemma}

We now describe a construction of an LDPC code using different powers of $f$. We consider parity check matrices of the form
 $$H=\begin{pmatrix}1 & 1 &  \ldots & 1 \\ 
1 & f^{A_1i_1} &  \ldots & f^{A_1i_r}   \\ 
1 & f^{A_2i_1} &  \ldots & f^{A_2i_r}   \\ 
\vdots & \vdots &  \ddots & \vdots \\
1 & f^{A_Ji_1} &  \ldots & f^{A_Ji_r}   \\ 
\end{pmatrix}.$$

Each of the numbers $A_k$ and $i_k$ is considered modulo $m$, since they appear as exponents of $f$. We assume that the set of numbers $A_k \in \Z_m$ are nonzero and pairwise different (mod $m$). Similarly, the numbers $i_k \in \Z_m$ are assumed to be nonzero and pairwise different. 

 By Theorem 4, $H$ has a maximum girth of 12. Write $$\mathcal{I} = \{0, i_1, i_2, \ldots, i_r\} \subseteq \Z_m, \ i_1 < i_2< \ldots < i_r,$$  $$\mathcal{A}=\{0, A_1, A_2, \ldots A_J\}\subseteq \Z_m.$$

Note that any $2t$-cycle has a trajectory that can be represented as 
\begin{align*}
    a_{t}j_1\rightarrow  {a_1j_1}\rightarrow a_1j_2\rightarrow {a_2j_2} \rightarrow  a_2j_3\rightarrow  \ldots  \rightarrow {a_i j_{i+1}}\rightarrow a_ij_{i+2} \rightarrow \ldots  a_{t}j_{t} \rightarrow a_{t}j_1
\end{align*}
 where $a_k \in \mathcal{A}$, $j_k \in \mathcal{I}$, $1\leq k\leq t$, $a_i \neq a_{i+1},\ j_i \neq j_{i+1}, \ 1\leq i\leq t-1$,  $a_{t} \neq a_1$, $j_t \neq j_1$.
Denote by $c_{2t}$ the column index of the $(2t)^{th}$ term of the $2t$-cycle. Then, we have $c_{2t}=f^{i} (c_0)$ where \begin{align*}
    i & =(j_{t}-j_1)a_{t} + (j_1-j_2)a_1+\ldots + (j_{i}-j_{i+1})a_{i} + \ldots +  (j_{t-1}-j_{t})a_{t-1} \\
    & = (a_1-a_{t})j_1 + (a_2-a_1)j_2+\ldots +  (a_{i+1}-a_{i})j_{i}  + \ldots + (a_{t}-a_{t-1})j_{t} 
\end{align*}

We revisit 12-cycles and demonstrate Theorem~\ref{12c} in this setting. For 12-cycles, $$i = (a_0-a_5)j_1 + (a_1-a_0)j_2+\ldots + (a_5-a_4)j_6.$$ Set $j_1=j_4= 0$, $j_2=j_5 = j$, $j_3=j_6 = k$, $a_0=a_2=a_4$, $a_1=a_3=a_5$. Then $i = (a_1-a_0)(j-k) + (a_0-a_1)(-k+j)=0$. Thus, $H$ always has a 12-cycle. 

The following lemma describes a relationship between $2k$-cycles in $H$, and the sets $\mathcal{A}_D^k$, $\mathcal{I}_S^k$, $\mathcal{I}_D^k$, and $\mathcal{I}_S^k$. 
\begin{lemma} The following statements are equivalent.
\begin{enumerate}
\item $H$ has a $2k$-cycle.
    \item There exists a sequence $A_d =(b_1, \ldots, b_{k}) \in \mathcal{A}_{D}^{k}$, a sequence $J_s=(j_1\ldots, j_{k})\in \mathcal{I}_{S}^{k} $, and a point $c_0 \in \Z_n$ such that for all $k'<k$, $f^{b_1j_1 + \ldots + b_{k'} j_{k'}}(c_0)\neq c_0$, and $f^{b_1j_1 + \ldots + b_{k} j_{k}}(c_0)=c_0$. 
    
    \item There exists a sequence $A_S =(b_1, \ldots, b_{k}) \in \mathcal{A}_{S}^{k}$, a sequence $J_d=(j_1\ldots, j_{k})\in \mathcal{I}_{d}^{k} $, and a point $c_0 \in \Z_n$ such that for all $k'<k$, $f^{b_1j_1 + \ldots + b_{k'} j_{k'}}(c_0)\neq c_0$, and $f^{b_1j_1 + \ldots + b_{k} j_{k}}(c_0)=c_0$. 
\end{enumerate}
\end{lemma}

\begin{remark}
Note that the roles of $\mathcal{I}$ and $\mathcal{A}$ are symmetric in the consideration of cycles. This is to be expected, since in general taking the transpose of the parity check matrix does not alter the cycles. Thus, we will often not specify conditions for both the symmetric cases. 
\end{remark}

Note that in the case when $f$ is a prime $m$-cycle, $f^i$ has a fixed point if and only if $i=0 \mod m$. Thus, we have

\begin{lemma} If $f$ is an $m$-cycle and $m$ is prime, then $H$ has no $2k$-cycle if and only if for any sequences $A_d =(b_1, \ldots, b_{k}) \in \mathcal{A}_{D}^{k}$, $J_s=(j_1\ldots, j_{k})\in \mathcal{I}_{S}^{k} $, we have ${b_1j_1 + \ldots + b_{k} j_{k}}\neq 0$. 
\end{lemma}

We already know that every parity check matrix $H$ of the above form contains a 12-cycle. Thus, to determine the girth of $H$, we find conditions for 4-, 6-, and 8-cycles. The next section will focus on deriving these conditions in terms of properties of the sets $\mathcal{A}$ and $\mathcal{I}$.

\section{Conditions for avoiding $4$-, $6$-, and $8$-cycles}\label{cycles}

This section is divided into two cases, based on the number of block-rows in $H$, namely the cases $L=2$ and $L>2$. In the former case, the set $\mathcal{A}$ can be taken to be equal to $\{0,1\}$ without loss of generality. In this case, we obtain equivalent conditions for $H $ to be free from $4$- and $8$-cycles in terms of conditions on the set $\mathcal{I}$. In the latter case, where $L>2$, we obtain an equivalent condition for $H$ to be free from 4- and 6-cycles in terms of conditions on the sets $\mathcal{A}$ and $\mathcal{I}$, and give a necessary condition for the 8-cycle case.

\subsection{Case $L=2$}

We observe that for $L=2$, all $(4k+2)$-cycles are absent. To see this, note that in this case, any such cycle has an associated sequence of the form $$(j_0,l_0), (j_1, l_0), (j_1, l_1), (j_2, l_1)\ldots, (j_{2k+1}, l_{2k}) $$ where $j_{2i}=j_0, j_{2i-1}=j_1$ for all $i$, where $j_1\neq j_0$. Thus, $j_{2k+1} \neq j_0$, and there is no cycle of length $4k+2$.

We first illustrate the trivial subcase $J=2$, which corresponds to very small-dimensional codes, and so is of no interest in coding theory. Moreover, this case is not really practical since in this case, $H$ has both column and row weight equal to $2$ and the same is true for any submatrix of $H$ corresponding to a cycle. This implies that any cycle of length $4k$ corresponds to a codeword of weight $2k$ and therefore the minimum distance is upper bounded by $2k$. However, we want to start with this case for illustrative purposes.

Note that in this case, the proof of Theorem~\ref{12c} does not work, and in general a 12-cycle may not exist. Here, a general condition can be given for any $4k$-cycle. 

\begin{theorem}
Let $f: \Z_n \rightarrow \Z_n$ be a permutation and consider the parity check matrix given by 

$H =\begin{pmatrix}1 & 1 \\ 1 & f 
\end{pmatrix}$. $H$ has no $4k$-cycles if and only if $f^{k}$ is a derangement. In particular, let $f$ be a derangement and $r>2$ be the smallest prime divisor of the order of $f$. Then, $g(H) \geq 4r$. Thus if $f$ is a derangement with odd order, then $g(H)\geq 12$. If $f$ has prime order $m$, then $g(H)=4m$.
\end{theorem}
\begin{proof}
Let $k\geq 1$. By Theorem 1 $H$ has a $2k$-cycle
if and only if there exists a column index $c_0$, $0\leq c_0 \leq l_0$, 
such that for $c_{i+1} = (R_{j_{i+1}, l_{i+1}}^{-1}(R_{j_{i+1}, l_i}(c_{i}))$, $1 \leq i \leq k-1$, we have $c_0 = c_{2k}$. It is easy to see that here, $c_{i+1} = f^{-1}(c_i)$ if $i$ is even, and $c_{i+1} = c_i$ if $i$ is odd. 
By induction, we have $c_{2k} = f^{-k}(c_0)$. Thus, a cycle of length $4k$ exists if and only if for some $c_0$ we have $f^{-k}(c_0) =c_0$, i.e. if and only if $f^{k}(c_0)$ has a fixed point. 
\end{proof}



 
 We now turn to the more interesting subcase, where $L=2$, and $J>2$. Here, the parity check matrix can be written as $H=\begin{pmatrix}1 & 1 &  \ldots & 1 \\ 
1 & f^{i_1} &  \ldots & f^{i_r}   \\ 
\end{pmatrix}$, where $r\geq 2$. As before, write $\mathcal{I} = \{0, i_1, \ldots, i_r\}$.

The following result shows that for a fixed permutation $f$, 4-cycles in $H$ are determined entirely by the set $\mathcal{I}_\Delta$.
     \begin{theorem}\label{4cyc} For $r>1$ the parity check matrix $H=\begin{pmatrix}1 & 1 &  \ldots & 1 \\ 
1 & f^{i_1} &  \ldots & f^{i_r}   \\ 
\end{pmatrix}$
has no $4$-cycle if and only if $f^{t}$ is a derangement for every $t \in \mathcal{I}_\Delta$. 
\end{theorem}
\begin{proof}
Up to the order of rows, a 4-cycle has a trajectory of the form $1\rightarrow  {j_1}\rightarrow j_2\rightarrow 1 \rightarrow 1$, for $j_1, j_2 \in \mathcal{I}$ so we have $c_4 = f^{(j_1-j_2)}(c_0)$ for some distinct $j_1, j_2 \in \mathcal{I}$. Thus, there is a one-to-one correspondence between 4-cycles in $H$ and elements of $\mathcal{I}_\Delta$, and $H$ has no $4$-cycle if and only if $f^{i}$ is a derangement for every $i \in \mathcal{I}_\Delta$.
\end{proof}

 \begin{corollary}\label{4cyc-prime}
 Suppose that $f$ is a derangement of prime order $m=n$. Then $H$ has no 4-cycle.
 \end{corollary}

 
 Thus this technique allows us to construct 4-cycle free $k \times n$ parity-check matrices, as illustrated below. 
 
 \begin{example}
 We take $\mathcal{I}=\{0,1,4,6,12,10,15, 24 \}$, so  \begin{align*} \mathcal{I}_\Delta=\{\pm 1, \pm 2, \pm 3, \pm 4, \pm 5, \pm 6,\pm 8,\pm 9, \pm 10, \pm 11, \pm 12,  \pm 14, \pm 15,   \pm 18, \pm 20, \pm 23, \pm 24\}\end{align*}
 By Corollary \ref{4cyc-prime}, choosing $f$ to be any permutation of order coprime to $2,3,5,7,11,23$, we obtain a parity check matrix free from $4$-cycles. Of course, one may simply choose $f$ to be a cycle of prime order, say $f = (1 \ 2 \ 3 \  \ldots \ 17)$, which gives a $(136, 103)$-code. However, we may also construct a $(208, 158)$-code with no 4-cycles using a non-cycle permutation,  \begin{align*} f=&(1 \ 2 \ \ldots \ 12 \ 13) \cdot (14 \ 15 \ldots \ 24 \ 25\ 26).\end{align*}
 \end{example}

The above results demonstrate the superiority of using derangements of prime order $m$, since for these, the only condition on $\mathcal{I}$ is that $i_r <m$, so $\mathcal{I}$ attains the largest possible size. In general, if $f$ has order $m$, for no $4$-cycles one requires that all prime factors of $m$ lie outside $\mathcal{I}_\Delta$.

 As discussed above, the case $L=2$ is always free from 6-cycles (in general, from $4k+2$ cycles). The below result characterizes the criterion for 8-cycles. 

 \begin{theorem}\label{injadd} Let $L=2, J>2$. Write $\mathcal{I} = \{i_1, i_2, \ldots, i_r\} \subseteq \Z_m$, with $i_1 < i_2< \ldots < i_r$. 
 Then the parity check matrix
 $H=\begin{pmatrix}1 & 1 &  \ldots & 1 \\
1 & f^{i_1} &  \ldots & f^{i_r}   \\
\end{pmatrix}$
 has no $8$-cycle if and only if $f^{i}$ is a derangement for every $i \in (\mathcal{I}_+)_\Delta$. In particular, if $H$ has no 8-cycle then $\mathcal{I}$ is a $B_2$-set.

\end{theorem}
\begin{proof}
An 8-cycle can be represented by a trajectory of the form $a_{4}j_1\rightarrow  {a_1j_1}\rightarrow a_1j_2\rightarrow  \ldots \rightarrow a_{4}j_1$. The condition for an 8-cycle is $c_0 = f^{i}(c_0)$ for some $c_0 \in \Z_n$, where $i=[(i_2+i_4)-(i_1+i_3)]$ for $\{i_2,i_4\}\neq \{i_1, i_3\}$, $i_k \in \mathcal{I}, \ 1\leq i\leq 4$. In other words, $i \in (\mathcal{I}_+)_\Delta$. The result is now clear.
\end{proof}


With an additional condition, the above necessary result also can be made sufficient. 
\begin{corollary}\label{8cyc-prime}
If $L=2, \ J>2$, $m=n$ is prime and $f$ is an $m$-cycle, then $H$ has no $8$-cycles if and only if $\mathcal{I}$ is a $B_2$-set.
\end{corollary}
\begin{proof}
If $m=n$ is prime and $f$ is an $m$-cycle, then $f^i$ is a derangement if and only if $i=0 \mod m$. Thus, $H$ has an 8-cycle if and only if $0 \in  ({\mathcal{I}_+})_\Delta $, which, in turn, happens if and only if  $\mathcal{I}$ is not a $B_2$-set by Lemma~\ref{equivconds}.
\end{proof}


From Corollaries~\ref{4cyc-prime} and \ref{8cyc-prime} we have

\begin{corollary}
In the case $L=2, \ J>2$, if $f$ is an $m$-cycle, and $n=m$ is prime, we have \begin{align*} \mathcal{I} \text{ is an $m$-$B_2$-set } \implies g(H) = 12.\end{align*} 
\end{corollary}


\begin{example}
Consider $\mathcal{I}= \{0, 1,4,6,13\}$. Then we have  $\mathcal{I}+= \{0,1,2,4,5,6,7,8,10,12, 13,14,17,19,26\}$. By Lemma~\ref{equivconds}, $\mathcal{I}$ is a $B_2$-set (over $\Z$). Further, we have $\mathcal{I}_\Delta= \{\pm 1, \pm 2, \pm 3, \pm 4, \pm 5, \pm 6, \pm 7, \pm 9, \pm 12, \pm 13 \}$. Thus, choosing any $m$ coprime to $2,3,5,7,13,17,19$, we get that $0 \pmod m \not \in \mathcal{I}_\Delta$  and $\mathcal{I}$ is an $m$-$B_2$ set.
As a concrete example, $f$ can be taken to be the 29-cycle, $f=(1 \ 2 \ 3 \ 4 \ \ldots \ 29)$. Then the parity check matrix $H=\begin{pmatrix}1 & 1 &  1 & 1 &1 \\
1 & f^{1} &  f^4 & f^6 & f^{13}   \\
\end{pmatrix}$ has girth $g(H)=12$. This gives a $(145, 88)$-code with girth 12.
\end{example}

 \subsection{Case $L>2$}
 
 In this section, we study the case $L>2$. By symmetry, and since we have already dealt with the case $L=2, J>2$, we may also assume that $J>2$. Here, the parity check matrix is given by $H=\begin{pmatrix}1 & 1 &  \ldots & 1 \\ 
1 & f^{A_1i_1} &  \ldots & f^{A_1i_r}   \\ 
1 & f^{A_2i_1} &  \ldots & f^{A_2i_r}   \\ 
\vdots & \vdots &  \ddots & \vdots \\
1 & f^{A_Ji_1} &  \ldots & f^{A_Ji_r}   \\ 
\end{pmatrix}$, $L>2$. We begin by obtaining a condition for $4$-cycles, which can easily be seen as a generalization of Theorem~\ref{4cyc}.
 
 \begin{theorem}\label{4cyc-lg2}
 $H$ contains no 4-cycle if and only if $f^i$ is a derangement for all $i \in \mathcal{A}_\Delta\cdot \mathcal{I}_\Delta$.
 \end{theorem}
 \begin{proof}
 Note that a 4-cycle can in general be represented by the trajectory $a_{2}j_1\rightarrow  {a_1j_1}\rightarrow a_1j_2\rightarrow {a_2j_2} \rightarrow a_2j_1$, or by the sequence $(c_0,j_0)\rightarrow (r_0,c_0)\rightarrow (r_1,c_1)\rightarrow (r_1, c_1)\rightarrow (r_1, c_0)$ where the column index $c_2$ is given by $c_2 = f^i(c_0)$ with $i = (a_1-a_2)(j_1-j_2) $, so $i \neq 0$ since $j_1 \neq j_2$, $a_0 \neq a_1$, and $i \in \mathcal{A}_\Delta\cdot \mathcal{I}_\Delta$. Conversely, given any $i \in \mathcal{A}_\Delta\cdot \mathcal{I}_\Delta$, one can write $i$ in the above form and associate a 4-cycle to it if $f^i$ is not a derangement.
 \end{proof}
 
 \begin{corollary}
 If 
  $f$ is a prime $m$-cycle, then $g(H) \geq 6$.
 \end{corollary}
 \begin{proof}
 We always have $0 \not \in \mathcal{A}_\Delta$ and $0 \not \in \mathcal{I}_\Delta$, so by the primeness of $m$, $0 \not \in \mathcal{A}_\Delta\cdot \mathcal{I}_\Delta$. Thus for each $i \in \mathcal{A}_\Delta\cdot \mathcal{I}_\Delta$, $i$ is coprime to $m$ and so $f^i$ is a derangement. Thus by Theorem~\ref{4cyc-lg2}, $H$ has no 4-cycles.
 \end{proof}

\begin{example}\label{ex3}
Take $\mathcal{A}=\{0,1,-1\}$, so $\mathcal{A}_\Delta= \{-2,-1,1,2\}$. Here, $H=\begin{pmatrix}1 & 1 &  \ldots & 1 \\ 
1 & f^{i_1} &  \ldots & f^{i_r}   \\ 
 1& f^{-i_1}  &  \ldots & f^{-i_r}  \end{pmatrix}$. By Theorem~\ref{4cyc-lg2}, $H$ has no 4-cycle if and only if $f^i$ is a derangement for all $i \in \mathcal{I}\cup 2\mathcal{I}\cup \mathcal{I}_\Delta \cup 2\mathcal{I}_\Delta$. Take $\mathcal{I}=\{0,2,4,6,8\}$. Then $\mathcal{I}_\Delta=\{2,4,6,8, -2,-4,-6,-8\}$, $2\mathcal{I}=\{4,8,12,16\}$, $2\mathcal{I}_\Delta=\{4,6,8,12, -4,-8,-12,-16\}$. Taking $m$ to be a prime larger than $16$ and $f$ to be an $m$-cycle, we get $g(H)\geq 6$. E.g. with $f = (1 \ 2 \ 3 \  \ldots \ 17)$, $H$ is a $(85, 52)$-code.  
 \end{example}

Note that when $L>2$, 6-cycles (more generally, $2k$-cycles for $k$ odd) may occur in $H$. We now present an equivalent condition for 6-cycles.

 \begin{theorem}
 If $L, J>2$, then $H$ contains no 6-cycles if and only if $\mathcal{I}$ and $\mathcal{A}$ have the following property: For three distinct $j_1, j_2, j_3 \in \mathcal{I}$ and any $c_1, c_2 \in \mathcal{A}_\Delta$ such that $(c_1+c_2) \in \mathcal{I}_\Delta$, we have $f^{c_1j_1+c_2j_2 - (c_1+c_2)j_3}$ is a derangement.  In particular, if $f$ is a prime $m$-cycle and for any $c_1, c_2 \in\mathcal{A}_\Delta$ such that $(c_1+c_2) \in \mathcal{A}_\Delta$, and any three distinct $j_1, j_2, j_3 \in \mathcal{I}$, we have ${c_1j_1+c_2j_2 \neq (c_1+c_2)j_3} \mod m$, then $H$ contains no 6-cycles.
 \end{theorem}
 \begin{proof}
  For  6-cycles, $c_3 =f^i(c_0)$ with $$i = (a_0-a_2)j_1 +(a_1-a_0)j_2+ (a_2-a_1)j_3 $$ with $j_1\neq j_2, j_2\neq j_3, j_3 \neq j_1$. Clearly this case requires $J>2$, and $a_0 \neq a_1, a_0 \neq a_2, a_1\neq a_2$. We can write $i= c_1j_1+ c_2j_2+ (-c_1-c_2)j_3$, with $c_1, c_2, c_1+c_2 \in \mathcal{I}_\Delta$. Conversely, $c_1, c_2 \in\mathcal{A}_\Delta$ such that $(c_1+c_2) \in \mathcal{A}_\Delta$ and  ${c_1j_1+c_2j_2 = (c_1+c_2)j_3} \mod m$ for any three distinct $j_1, j_2, j_3 \in \mathcal{I}$ leads to $i=0$, i.e. to a $6$-cycle. The result is now clear. 
 \end{proof}

We now turn our attention to 8-cycles. Recall that any $8$-cycle has a trajectory that can be represented as 
\begin{align*}\label{2t-cycle}
    a_{4}j_1\rightarrow  {a_1j_1}\rightarrow a_1j_2\rightarrow \ldots  \rightarrow  
     a_{4}j_{4} \rightarrow a_{4}j_1
\end{align*}
 where $a_k \in \mathcal{A}$, $j_k \in \mathcal{I}$, $a_i \neq a_{i+1},\ j_i \neq j_{i+1}, \ 1\leq i\leq t-1$,  $a_{t} \neq a_1$, $j_t \neq j_1$.
 Thus, we have $c_{4}=f^{i} (c_0)$ where \begin{align*}
    i =&(j_{4}-j_1)a_{4} + (j_1-j_2)a_1+ (j_2-j_3)a_2+ (j_{3}-j_{4})a_3 \\ = &(a_1-a_{4})j_1 + (a_2-a_1)j_2+(a_3-a_2)j_3 + (a_4-a_3)j_{4} 
\end{align*} and $c_{4}$ denotes the column index of the $8^{th}$ term of the $8$-cycle. 

In this case, we have the following necessary condition.
 \begin{theorem}\label{8cyclg2}
 If $L>2$ and $\mathcal{A}$ or $\mathcal{I}$ is not an $m$-$\mathcal{B}_2$-set then $H$ contains an 8-cycle.
 \end{theorem}
 \begin{proof}
 As above, we have, for  8-cycles, $$i = (a_0-a_3)j_1 +(a_1-a_0)j_2+ (a_2-a_1)j_3+(a_3-a_2)j_4.$$  First suppose that $\mathcal{A}$ is not an $m$-$B_2$ set. Then there exists a sequence $(a_0,a_1,a_2,a_3)$ such that $a_0+a_2 = a_1+a_3$, with $\{a_0,a_2\} \neq \{a_1, a_3\}$. It is easy to see that we must have $a_0 \neq a_1, a_1 \neq a_2, a_2 \neq a_3, a_3 \neq a_0$ here, so this sequence lies in $\mathcal{A}_d^4$. Set $j_1=j_3$, $j_2=j_4$. Then $$i = (a_0-a_1+a_2-a_3)(j_1-j_2)=0.$$
 Thus, in this case $H$ always has an 8-cycle. 
 
  Similarly, suppose that $\mathcal{I}$ is not an $m$-$B_2$ set. Then there exist $j_1, j_2, j_3, j_4 \in I$ such that $j_1+j_3 = j_2+j_4$, with $\{j_1,j_3\} \neq \{j_2, j_4\}$. Again, it is easy to see that we must have $j_1 \neq j_3, j_2 \neq j_4$.  Thus, this sequence lies in $\mathcal{I}_d^4$.  Set $a_1=a_3$, $a_0=a_2$, so that $i = (j_1-j_2+j_3-j_4)(a_0-a_1)$ and $
      a_{1}j_1\rightarrow  {a_0j_1}\rightarrow a_0j_2\rightarrow {a_1j_2} \rightarrow a_1j_1\rightarrow a_0j_1\rightarrow  a_0j_2 \rightarrow a_{1}j_{2} \rightarrow a_{1}j_1
  $ defines an 8-cycle.
 \end{proof}
 
 \begin{example}
Again, take $\mathcal{A}=\{0,1,-1\}$, so $\mathcal{A}_\Delta= \{-2,-1,1,2\}$. Here, $H=\begin{pmatrix}1 & 1 &  \ldots & 1 \\ 
1 & f^{i_1} &  \ldots & f^{i_r}   \\ 
 1& f^{-i_1}  &  \ldots & f^{-i_r}   \\\end{pmatrix}$. Let $f$ be a prime $m$-cycle. Write $c_1 = a_0-a_2, \ c_2 = a_1-a_0, \ c_3 = a_2-a_1=-(c_1+c_2)$. All possibilities for these constants are listed below. \begin{table}[H]
    \centering
    \begin{tabular}{c|c|c|c|c|c}
     $a_0$ & 1 & -1 &  1 & -1 \\
     $a_1$  & 0 & 0 &  -1 & 1\\
     $a_2$ & -1 & 1 &  0 &0 \\
     \hline
     $c_1$& 2 & -2 &  1 & -1\\
     $c_2$  & -1 & 1  & -2& 2 \\
     $c_3$  & -1 & 1  &1& -1\\
    \end{tabular}
    \caption{Caption}
    \label{tab:my_label}
\end{table}

Note that here, $\{c_1,c_2,c_3\} =\{1,1,-2\}$ or $\{-1,-1,2\}$. Further, $(c_1,c_2)\in \{(1,1),(-1,-1), (-1,2), (1,-2) \}$ are the only possibilities of subsets of $\mathcal{A}_\Delta$ such that $c_1+c_2 \in \mathcal{A}_\Delta$.  By Theorem 8, and by the primeness of $m$, $H$ has no 6-cycles if and only if $\mathcal{I} = \{0,i_1, \ldots, i_r\}$ satisfies the following property: for distinct $j_1,j_2,j_3 \in \mathcal{I}$, we have $2j_1 = j_2+j_3 \iff j_1=j_2=j_3$. In loose terms, $\mathcal{I}$ does not contain the averages of its terms.

For instance, take $\mathcal{I}=\{0, 1, 4, 6, 10\}$, $m=n=17$. 
By Corollary~\ref{4cyc-prime}, $H$ has no 4-cycles. It is also clear that for all distinct $j_1, j_2, j_3 \in \mathcal{I}$, we have $j_1+j_2 \neq 2j_3$. Thus, $H=\begin{pmatrix}1 & 1 &  1 & 1& 1 \\ 
1 & f &  f^4 & f^6 & f^8  \\ 
 1& f^{-1}  &  f^{-4} &  f^{-6} & f^{-8}   \\\end{pmatrix}$ has no 4- or 6-cycles. However, by Theorem~\ref{8cyclg2}, it does have 8-cycles, since $\mathcal{A}$ is not an $m$-$B_2$ set, so it has girth $g(H) = 8$.

However, with $\mathcal{I}=\{0, 1, 2, 3, 4\}$, we have $2\cdot 2 = 0+4$, so $H=\begin{pmatrix}1 & 1 &  1 & 1& 1 \\ 
1 & f^1 &  f^2 & f^3 & f^4  \\ 
 1& f^{-1}  &  f^{-2} &  f^{-3} & f^{-4}   \\\end{pmatrix}$ has 6-cycles. From Example~\ref{ex3}, with $m=17$ and $f$ an $m$-cycle, $H$ is free from 4-cycles. Thus, $g(H)= 6$. 
\end{example}

\section{Extension to irregular LDPC codes}\label{irregular}

As observed e.g. in \cite{irregular} irregular LDPC codes can outperform regular LDPC codes. In this section, we describe how our previous considerations can be easily modified in order to obtain irregular LDPC codes.

First, consider again a parity-check matrix in the general form $$H=\begin{pmatrix}1 & 1 &  \ldots & 1 \\ 
1 & f^{A_1i_1} &  \ldots & f^{A_1i_r}   \\ 
1 & f^{A_2i_1} &  \ldots & f^{A_2i_r}   \\ 
\vdots & \vdots &  \ddots & \vdots \\
1 & f^{A_Ji_1} &  \ldots & f^{A_Ji_r}   \\ 
\end{pmatrix}.$$
To obtain a parity-check matrix of an irregular code extend $H$ to a new parity-check matrix
$$\hat{H}:=[H\ h_{r+1}\ \cdots\ h_L]$$
by adding additional blocks of columns of the form $h_s:=\begin{pmatrix} 0\\ \vdots\\ 0\\ f^{A_ki_s}\\ 0\\ \vdots\\ 0 \end{pmatrix}$, where $0$ represents the all-zero map on $\{1,\hdots,n\}$, $s\in\{r+1,\hdots,L\}$ and for each $s$ the parameter $k$ can be any value from $\{0,\hdots,J\}$ determining the position of the nonzero block in $h_s$, where we set $A_0:=0$. One can add arbitrarily many blocks of columns of this form without changing the number and sizes of cycles in the corresponding Tanner graph. Moreover, it is possible to increase the rate of the code with the help of these additional columns as illustrated in the following examples.

\begin{example}
Consider $m=n=17$ and let $f$ be the cyclic shift by one position, i.e. $f:\mathbb Z_{17}\rightarrow\mathbb Z_{17},\ x\mapsto x+1$. Then, $id+f$ corresponds to the mapping $x\mapsto 2x+1$ and is therefore an element of $\mathcal{S}_{17}$. It follows that the matrix $\begin{pmatrix} I & I\\ I & f\end{pmatrix}$ is invertible and therefore the matrix  $\begin{pmatrix} I & I & I & \hdots & I\\ I & f & f^{i_2} & \hdots & f^{i_r}\end{pmatrix}$ has full rank $34$ for any choices of $i_2,\hdots, i_r$. Adding any blocks of columns of the form $\begin{pmatrix} 0\\ f^{i_s}\end{pmatrix}$ or $\begin{pmatrix}  f^{i_t}\\ 0\end{pmatrix}$ does neither change the rank nor the cycle structure of the Tanner graph of $H$. However, the length of the code increases by $17$ with each added block and its rate is equal to $1-\frac{2}{r+1+a}$, where $a$ is the number of additional blocks of columns. In this way, one can achieve rates that are arbitrarily close to $1$ without obtaining additional cycles.
\end{example}

\begin{example}
Again, take $H=\begin{pmatrix}1 & 1 &  1 & 1& 1 \\ 
1 & f &  f^4 & f^6 & f^8  \\ 
 1& f^{-1}  &  f^{-4} &  f^{-6} & f^{-8}   \\\end{pmatrix}$. This is a $39 \times 65$ matrix of rank $37$ and girth $8$, defining a $(65, 28)$-linear code over $\F_2$ with rate $\approx 0.431$. Now consider an irregular code defined by extended $H$ as follows $$M=\begin{pmatrix}1 & 1 &  1 & 1& 1 & 0 & 0 & 1\\ 
1 & f &  f^4 & f^6 & f^8 & 0 & f & 0 \\ 
 1& f^{-1}  &  f^{-4} &  f^{-6} & f^{-8} & f^{-4} & 0 & 0   \\\end{pmatrix}.$$ Here $M$ is a $39 \times 104$ parity-check matrix of rank $39$ and girth $8$, defining a $(104, 65)$-linear code over $\F_2$ with rate $0.625$.
\end{example}

\section{Conclusion}\label{concl}

We provided a framework to link conditions for cycles in the Tanner graph of an LDPC code constructed from permutation matrices with number theoretic concepts like Sidon sets. To be able to avoid also larger cycles one should consider parity-check matrices built from non-commuting permutations. We leave this problem for future research.


\section*{Acknowledgements}
Julia Lieb acknowledges the support of Swiss National Science Foundation grant n. 188430. The work of Simran Tinani is supported by armasuisse Science and Technology. 
The authors thank Roxana Smarandache for her helpful feedback and comments.

\bibliographystyle{plain}
\bibliography{references}
\end{document}